\documentclass[copyright,creativecommons]{eptcs}
\providecommand{\event}{GandALF 2016} 

\title{Alternation Is Strict For Higher-Order Modal Fixpoint Logic}
\author{Florian Bruse
\institute{Universit\"at Kassel\\ Kassel, Germany}
\email{florian.bruse@uni-kassel.de}
}

\def\titlerunning{Alternation Is Strict For HFL}
\def\authorrunning{F. Bruse}

\usepackage{amssymb}
\usepackage{graphicx}
\usepackage{tikz}

\usepackage{xspace}
\usepackage{bussproofs}
\usepackage{stmaryrd} 
\usepackage{mathtools}
\usepackage{mathpartir}
\usepackage{amsthm}

\newcommand{\Transsys}{\ensuremath{\mathcal{T}}\xspace}
\newcommand{\hlmu}{\texorpdfstring{\ensuremath{\operatorname{HFL}}\xspace}{HFL}}

\newcommand{\flc}{\ensuremath{\operatorname{FLC}}\xspace}

\newcommand{\apkm}{\texorpdfstring{\ensuremath{\operatorname{APKA}}\xspace}{APKA}}
\newcommand{\aapkm}{\texorpdfstring{\ensuremath{\operatorname{APKA}}\xspace}{APKA}}
\newcommand{\verifier}{\ensuremath{\exists}\xspace}
\newcommand{\spoiler}{\ensuremath{\forall}\xspace}
\newcommand{\Prop}{\operatorname{Pr}}

\newcommand{\true}{\top}

\theoremstyle{plain}
\newtheorem{theorem}{Theorem}
\newtheorem{lemma}[theorem]{Lemma}
\newtheorem{definition}[theorem]{Definition}
\newtheorem{example}[theorem]{Example}
\newtheorem{remark}[theorem]{Remark}
\newtheorem{observation}[theorem]{Observation}
\newtheorem{corollary}[theorem]{Corollary}

\begin{document}

\maketitle

\begin{abstract}
We study the expressive power of Alternating Parity Krivine Automata (\apkm), which provide operational semantics to Higher-Order Modal Fixpoint Logic (\hlmu). \apkm consist of ordinary parity automata extended by a variation of the Krivine Abstract Machine. We show that the number and parity of priorities available to an \apkm form a proper hierarchy of expressive power as in the modal $\mu$-calculus. This also induces a strict alternation hierarchy on \hlmu. The proof follows Arnold's (1999) encoding of runs into trees and subsequent use of the Banach Fixpoint Theorem.
\end{abstract}

\section{Introduction}

Parity automata provide popular operational semantics for the modal $\mu$-calculus and, hence, for all regular properties over trees. They are equivalent to most other acceptance modes with the exception of B\"{u}chi automata \cite{rabin70}. However, since parity automata can only express regular properties, extending their expressive power, or extending them to cover stronger logics, is the subject of ongoing research. For example, visibly pushdown automata \cite{DBLP:conf/stoc/AlurM04} allow the addition of a limited pushdown stack but tie the stack operations to different and disjoint parts of the alphabet. 

In this paper, we revisit our previous work on extending parity automata by a variant of the Krivine Abstract Machine \cite{Krivine}, which incorporates a simply typed lambda calculus into the semantics of the automaton model. The resulting \emph{Alternating Parity Krivine Automata}  (\aapkm) yield operational semantics for Higher-Order Modal Fixpoint Logic (\hlmu) \cite{viswaviswa}. The acceptance condition of \aapkm is a stair parity condition over an acceptance game. The stair parity condition resembles that of visibly pushdown automata, but it is not tied to any alphabet symbols or tree labels, but rather emerges via the bookkeeping done by the Krivine Machine part. This automaton model is very expressive: Properties such as uniform inevitability or the presence of a given property in a level that is a power of two are easily expressible. This expressive power comes at a price, since emptiness of \aapkm, which is equivalent to satisfiability of \hlmu-formulae, is undecidable. 

A key improvement over the variant of \apkm presented in \cite{mfcs} is that in this paper, the state space of the automaton is not restricted to a tree-like structure inherited from \hlmu-formulae, but can take the form of any graph, just like an ordinary parity automaton is less restricted in structure than a formula of the modal $\mu$-calculus. Since in the new variant of \apkm, precedence between states representing different fixpoints can not be inferred from their position in a syntax tree, it is given explicitly via a parity labeling of states. This has the advantage that the alternation class of an automaton, or that of any equivalent formula, can be defined via the number of its priorities, while for formulae, alternation can be hard to gauge syntactically. Already for the modal $\mu$-calculus, syntactic criteria to define alternation classes can be quite complex \cite{DBLP:conf/lics/EmersonL86, DBLP:journals/tcs/Niwinski97}. On the automaton side of things, however, characterization via the number of priorities makes things much easier. Translations from \aapkm into \hlmu and vice versa are readily available and any alternation hierarchy for \apkm induces an alternation hierarchy on \hlmu. This settles the question posed in \cite{mfcs} on how to properly define alternation classes for \hlmu.

We find that for \aapkm, adding more priorities increases expressive power. 
 The original strictness result for parity automata has a beautiful proof \cite{arnold} involving the Banach Fixpoint Theorem, which also has been adapted to Fixpoint Logic with Chop \flc \cite{infcomp06}. Our strictness proof proceeds in a similar manner: Given an infinite binary tree and an \aapkm of suitable vocabulary, we construct another infinite binary tree which encodes the acceptance game of the run of the \aapkm. Given a vocabulary tailored to a specific alternation class, we construct an automaton which accepts such a game tree if and only if the original automaton accepts the original tree. This operation induces a contraction in the complete metric space of infinite binary trees, which, by the Banach Fixpoint Theorem, has a fixpoint. We show that on this fixpoint, no automaton with less priorities or with the same amount of priorities, but flipped parity, can be equivalent to the given meta-automaton. 

While strictness of the alternation hierarchy for \apkm and, hence, for \hlmu is not unexpected, such a result is not obvious. It is well known that adding more priorities to a parity automaton or more Rabin pairs to a Rabin automaton increases their expressive power, just as extra fixpoint alternation in the modal $\mu$-calculus does \cite{Bradfield}. However, adding extra fixpoint nesting does not always yield more expressive power: The Immerman-Vardi Theorem entails that, over finite ordered structures, first-order logic with least and greatest fixpoints  
is as strong as first-order logic with only one least fixpoint. Also the alternation hierarchy of the modal $\mu$-calculus itself collapses to the alternation-free fragment over certain classes of structures, for example the class of infinite words \cite{DBLP:conf/concur/Kaivola95} and, more generally, classes of structures with restricted connectivity \cite{DBLP:journals/tcs/GutierrezKL14}. Preliminary work also shows that alternation for \hlmu collapses over \emph{finite} structures. It should also be noted that, just like with Fixpoint Logic with Chop \cite{infcomp06}, formulas that are hard for alternation classes for the modal $\mu$-calculus are not necessary suitable candidates for higher-order logics. This is because these formulas are not designed for the higher-order features of \hlmu.

The plan of the paper is as follows: In Section~\ref{sec:apkm}, we define \aapkm and their acceptance condition for infinite binary trees. We have a look at their relation to \hlmu in Section~\ref{sec:hfl}. In the following section, we define alternation classes and present a class of trees that encode runs of \aapkm from a given alternation class. For each alternation class we also construct meta-automata that accept such a tree encoding a run if and only if the run was accepting. This allows us to prove strictness of the alternation hierarchy. The paper closes with a brief discussion of important points. 

\section{Alternating Parity Krivine Automata}
\label{sec:apkm}

Note that we previously defined \apkm differently. This work supersedes earlier definitions in \cite{mfcs}.
For ease of exposition, and since the alternation hierarchy argument is developed over the class of fully infinite binary trees, we only consider automata over labeled fully infinite binary trees. The concept of \aapkm extends naturally to trees of unrestricted branching factor and, or any class of Kripke structures

Fix some set $\mathcal{P}$ of propositions. An infinite binary tree with labels in $\mathcal{P}$ (just tree or $\mathcal{P}$-tree from now on) is given by a function $\mathcal{T}$ from the set $\{0,1\}^*$ of all $\{0,1\}$ words into $2^\mathcal{P}$. The root of the tree is identified with $\epsilon$ and the left and right successors of $t \in \{0,1\}^*$ are $t0$ and $t1$, respectively. We say that $P \in \mathcal{P}$ holds at $t$ (written $\mathcal{T},t\models P$) if $P \in \mathcal{T}(t)$. The pair $\mathcal{T},t$ refers to the subtree induced by $t$.

Simple types are defined inductively via
$\tau \Coloneqq \Prop \mid \tau\to\tau.$
We often refer to $\Prop$ as \emph{ground type}. The operator $\to$ is right-associative, so any type can be written as $\tau_1 \to \dotsb \to \tau_n \to \Prop$. The order $\operatorname{ord}$  is defined inductively via $ \operatorname{ord}(\Prop) = 0$ and $\operatorname{ord}(\tau_1 \to \dotsb \to \tau_n \to \Prop) = \max(\operatorname{ord}(\tau_1),\dotsc,\operatorname{ord}(\tau_n))+1$. The set of types is partially ordered via $\tau, \tau' < \tau \to \tau'$.
The intended semantics for the ground type over a tree $\mathcal{T}$ is a set of subtrees of $\mathcal{T}$, the intended semantics for a type of the form $\tau_1\to\dotsb\to\tau_n\to\Prop$ is that of a monotone function consuming arguments of types $\tau_1,\dotsc,\tau_n$ and returning a set of subtrees of $\mathcal{T}$.

\subsection{Definition}

\begin{figure}
\caption{Typing Rules for \aapkm-transition relations.}
\label{aapka-typing-rules}
\begin{mathpar}
\inferrule{ }{\Sigma \vdash P\colon\Prop} \and
\inferrule{ }{\Sigma \vdash \neg P\colon\Prop} \and
\inferrule{ }{\Sigma \vdash f^{i}_j\colon\tau^i_j} \and
\inferrule{ }{\Sigma \vdash X_i\colon\tau^i_1 \to \dotsb \to \tau^i_{n_{X_i}} \to \Prop} \and
\inferrule{\Sigma\vdash\varphi\colon\Prop}{\Sigma \vdash \Diamond_R \varphi \colon \Prop} \and
\inferrule{\Sigma\vdash\varphi\colon\Prop}{\Sigma \vdash \Box_R \varphi \colon \Prop} \and
\inferrule{\Sigma\vdash\varphi_1\colon\Prop \\  \Sigma\vdash\varphi_2\colon\Prop}{\Sigma \vdash \varphi_1 \vee \varphi_2 \colon \Prop} \and
\inferrule{\Sigma\vdash\varphi_1\colon\Prop \\  \Sigma\vdash\varphi_2\colon\Prop}{\Sigma \vdash \varphi_1 \wedge \varphi_2 \colon \Prop} \and
\inferrule{\Sigma\vdash \varphi\colon\tau \rightarrow \tau' \\ \Sigma\vdash\psi\colon\tau}{\Sigma \vdash (\varphi\,\psi) \colon \tau'} 

\end{mathpar}
\end{figure}

Fix a finite set $\mathcal{X} = \{X_1,\dotsc,X_n\}$ of states, or fixpoint variables, and a finite set $\mathcal{F}$ which is the disjoint union $\bigcup_n \mathcal{F}_{X_n}$ of lambda variables, where $\mathcal{F}_{X_i} =  \{f^{X_i}_1,\dotsc, f^{X_i}_{n_i}\}$.

For each $1 \leq i \leq n$ , let $\varphi_i$ be derived from the grammar
\[
 \varphi_i   \Coloneqq    P \mid \neg P \mid \Diamond \varphi_i \mid \Box \varphi_i \mid \varphi_i \vee \varphi_i \mid \varphi_i \wedge \varphi_i \mid f \mid X \mid (\varphi_i\,\varphi_i)  \]
 where $P \in \mathcal{P}$, $X \in \mathcal{X}$ and $f \in \mathcal{F}_{X_i}$. 
 

An \emph{Alternating Parity Krivine Automaton} (\aapkm) of index $m$ and order $k$ is a five-tuple the form 
$(\mathcal{X}, \Lambda, X_{\text{init}},\delta,(\tau_X)_{X \in \mathcal{X}})$
where $\mathcal{X}$ is as above, $X_{\text{init}} \in \mathcal{X}$ is the initial state, $\Lambda\colon\mathcal{X}\to\{1,\dotsc,m\}$ or  $\Lambda\colon\mathcal{X}\to\{0,\dotsc,m-1\}$ labels each fixpoint variable with a priority, the $\tau_X = \tau^X_1 \to \dotsb \to \tau^X_{n_X} \to \Prop$ of order at most $k$ specify the types of the fixpoints, the type of the initial state $\tau_{X_{\text{init}}}$ is $\Prop$, and $\delta$ is the transition relation that maps $X_i$ to $\varphi_i$ and is such that $\emptyset\vdash \delta(X)\colon\Prop$ 
for each $X \in \mathcal{X}$, according to the typing rules reproduced in Figure~\ref{aapka-typing-rules}. 
The state space of the automaton is $\mathcal{Q} = \mathcal{X} \cup \bigcup_{X \in \mathcal{X}} \operatorname{sub}(\delta(X))$, where $\operatorname{sub}(\psi)$ is the set of subformulae of $\psi$. 

\begin{example}
\label{example-inf-start}
Let $\mathcal{X} = \{I, X, Y\}$, let $\tau_I = \tau_Y = \Prop$, $\tau_X = \Prop\to\Prop$, let $\Lambda(I) = \Lambda(X) = 1, \Lambda(Y) = 0$. Let \begin{align*}
\delta(I) =&\, \emptyset &\mapsto&\, (X\, \neg P)  \\
\delta(X) =&\, x\colon \Prop &\mapsto&\, (\Diamond x) \vee \Box Y \\
\delta(Y) =&\, \emptyset &\mapsto&\, (X\,Y) 
\end{align*}
Let $\mathcal{A} = (\mathcal{X}, \Lambda, I, \delta, (\tau_X)_{X \in \mathcal{X}})$. We will see a run of $\mathcal{A}$ in Example~\ref{example-inf-cont}. This automaton corresponds to the \hlmu-formula (see Section~\ref{sec:hfl} for a definition of \hlmu) $\big(\mu X. \lambda x. \Diamond x \vee \Box \nu Y. (X\,Y)\big) \neg P$.
\end{example}

\subsection{Acceptance}

In the context of an \aapkm, an environment is either the empty environment $e_0$ or of the form $e  = (f^X_1 \mapsto (\psi_1, e_1),\dotsc,f^X_{n_X} \mapsto (\psi_{n_X},e_{n_X}),e')$ where the $\psi_i$ are in $\mathcal{Q}$, i.e., subformulae of $\delta(X)$ for some $X$. We call $e'$ the \emph{parent environment} of $e$, and any environment reachable via the irreflexive, transitive closure of this relation a \emph{predecessor} of $e$. 
A pair $(\psi_j, e_j)$ is called a closure.  
We set $e(f^X_i) = (\psi_i, e_i)$. While the set of environments never appears explicitely, we tacitly assume that at any point during a run of an \apkm, the only environments in existence are $e_0$ and any environments the automaton has created so far. This also means that all environments have only finitely many predecessors.

A configuration in a run of the automaton over some tree $\mathcal{T}$ has the form
$
(t,(Q,e), e', \Gamma, \Delta),
$
where $t$ is subtree of $\mathcal{T}$, $Q \in \mathcal{Q}$ is a subformula of $\delta(X)$ for some $X$,  $e$ and $e'$ are environments, $\Gamma$ is a possibly empty stack of closures, 
and $\Delta$ is a finite sequence of priorities. In each configuration, if the type of the current closure is $\tau_1 \mapsto \dotsb \tau_n \mapsto \Prop$ then there are $n$ elements on the stack, and their types are, from bottom to top, $\tau_1,\dotsc,\tau_n$. The latter invariant is by induction over the definition of the transition semantics.

A run over $\mathcal{T},t_0$ is a possibly infinite sequence of configurations that begins with the initial configuration $(t_0, (X_{\text{init}},e_0), e_0, \epsilon, \epsilon)$ and is produced by a two-player game between players \verifier and \spoiler. In each configuration, the next configuration is either produced deterministically, or one of the two players picks a successor. A run is accepting if \verifier wins the game according to a winning condition which we state later.

The transition semantics from $(t,(Q,e), e', \Gamma, \Delta)$ is as follows:
\begin{itemize}
\item If $Q$ is $X \in \mathcal{X}$, the automaton transitions towards $\delta(X)$. The closures on the stack are, from bottom to top, the closures $(\psi_{1},e''_{1}),\dotsc,(\psi_{n_X},e''_{n_X})$ of types $\tau_{1}^X,\dotsc,\tau_{n_X}^X$. 
 The automaton creates a new environment $e'' = (f^X_1 \mapsto (\psi_1, e''_1),\dotsc,f^X_{n_X} \mapsto (\psi_{n_X},e''_{n_X}), e')$, removes all these closures from the stack (which is now empty) and transitions to $(t,(\delta(X), e''), e'', \epsilon,\Delta')$, where  
$\Delta'$ is $\Delta$ with the priority of $Q$ appended.
\item If $Q$ is of the form $(\psi_1\, \psi_2)$, then the automaton pushes $(\psi_2,e)$ on the stack and transitions to the configuration $(t,(\psi_1,e),e',\Gamma \cdot (\psi_2,e), e', \Delta)$.
\item If $Q$ is of the form $f^X_{j}$ and not of type $\Prop$, then the automaton transitions to $(t,e(f^X_j), e', \Gamma,  \Delta)$.
\item If $Q$ is of the form $f^X_{j}$ and of type $\Prop$, and if $e(f^X_{j}) = (Q', e'')$ with $e' \not = e''$, then the automaton transitions to $(t,(Q,e), e'', \Gamma,  \Delta')$ where $e''$ is the parent of $e'$ and $\Delta'$ is $\Delta$ without the top element.
\item If $Q$ is of the form $f^X_{j}$, of type $\Prop$, and if $e(f^X_{j}) = (Q', e'')$ with $e' = e''$, then the automaton transitions to $(t,e(f^X_{j}), e', \Gamma,  \Delta)$.
\item If $Q$ is of the form $\psi_1 \vee \psi_2$ or $ \psi_1 \wedge \psi_2$ then the automaton transitions to $(t,(\psi_1,e),e',\Gamma,\Delta)$, respectively $(t,(\psi_2,e),e',\Gamma,\Delta)$, depending on \verifier's, respectively \spoiler's choice.
\item If $Q$ is of the form $\Diamond \varphi$ or $\Box \varphi$ then \verifier, respectively \spoiler, chooses a successor $t' \in \{t0,t1\}$ and the automaton transitions towards $(t',(\varphi,e),e',\Gamma, \Delta)$.\footnote{Over possibly finite trees, a player who is stuck loses the game.}
\item If $Q$ is of the form $P$ or $\neg P$ then \verifier wins if $\mathcal{T},t \models Q$ and \spoiler wins if $\mathcal{T},t\not\models Q$.
\end{itemize}
By induction, the transition relation alone determines the winner of all finite plays of the  game. The winner of an infinite play is determined by the behavior of the priority stack (see the end of this subsection).

Note that, in a departure from the usual way the Krivine Abstract Machine works, we insist that the equivalent of lambda abstraction pop the entire stack via a string of lambda abstractions implicit in each $\delta(X)$. While this is no proper restriction in expressive power, it makes bookkeeping which fixpoint is currently being computed much easier (see Definition~\ref{def:compute}).

Before we formalize the winner of an infinite play of the acceptance game, we illustrate the transition semantics via an example.

\begin{example}
\label{example-inf-cont}
Consider the infinite binary tree where only the first two levels are labeled by $P$. Since all subtrees on a level are isomorphic, we refer to the root as $r$ and all subtrees of level $i$ as $t_i$. An example run of the automaton $\mathcal{A}$ from Example~\ref{example-inf-start} over this tree is depicted in Figure~\ref{run-exml}. This example is adapted from \cite{infcomp06}. The highest priority that occurs infinitely often during the run is $1$ and, hence, odd. However,  all these occurrences of $1$ except the first two are eventually removed from the priority stack and the remaining priorities are all $0$. We will see later that this means that the automaton accepts.
\end{example}
\begin{figure}
\caption{Part of an example run of the \apkm from Example~\ref{example-inf-cont}.}
\label{run-exml}
\small
\begin{align*}
C_0 =&  (r,(I,e_0),e_0,\epsilon, \epsilon)                  & C_{10} =& (t_2,(x,e_3),e_3,\epsilon, 1101)                                \\
C_1 =&  (r,((X\,\neg P),e_0),e_0,\epsilon, 1)               & C_{11} =& (t_2,(x,e_3),e_2,\epsilon, 110)                        \\
C_2 =&  (r,(X,e_0),e_0,\neg P,   1)                         & C_{12} =& (t_2,(Y,e_2),e_2,\epsilon, 1101)                    \\
C_3 =&  (r,(((\Diamond x) \vee \Box Y),e_1),e_1,\epsilon, 11)  & C_{13} =& (t_2,((X\,Y),e_4),e_4,\epsilon, 1100)                        \\
C_4 =&  (r,((\Box Y),e_1),e_1,\epsilon, 11)                    & C_{14} =& (t_2, (X,e_4),e_4, (Y,e_4), 1100)                  \\
C_5 =&  (t_1,(Y,e_1),e_1,\epsilon, 11)                         & C_{15} =&  (t_2,(((\Diamond x) \vee \Box Y),e_5),e_5,\epsilon, 11001)     \\
C_6 =&  (t_1,((X\,Y),e_2),e_2,\epsilon, 110)             & C_{16} =&  (t_2,((\Diamond x),e_5),e_5,\epsilon, 11001)                         \\
C_7 =&  (t_1,(X,e_2),e_2,(Y,e_2), 110)                     & C_{17} =& (t_3,(x,e_5),e_5,\epsilon, 1101)                \\
C_8 =&  (t_1,(((\Diamond x) \vee \Box Y),e_3),e_3,\epsilon, 1101)   & C_{18} =& (t_3,(x,e_5),e_4,\epsilon, 1100)                        \\
C_9 =&  (t_1,((\Diamond x),e_3),e_3,\epsilon, 1101)            & C_{19} =& (t_3,(Y,e_4),e_4,\epsilon, 1100)                         \\                            
\end{align*}
\begin{align*}
e_1 &= (x\mapsto (\neg P, e_0), e_0) &  e_2 &= (\epsilon, e_1) &
e_3 &= (x \mapsto (Y, e_2),e_2) \\  e_4 &= (\epsilon, e_2) &  e_5 &= (x \mapsto (Y, e_4),e_4) & \\
\end{align*}
\end{figure}

\begin{definition}
\label{def:occurrence}
Let $C_i = (t_i,(\psi_i,e_i), e'_i, \Gamma_i,  \Delta_i)$ be a configuration. If $\psi_i = X$ then we say that $X$ \emph{occurs} in $C_i$. An \emph{occurrence} of a fixpoint variable is a configuration such that the variable occurs in that configuration. Moreover, $C_{i+1} = (t_{i}, (\delta(X),e_{i+1}), e_{i+1}, \epsilon,  \Delta_i \cdot \Lambda(X))$ is such that $e_{i+1}$ is new and there is a new priority on top of the priority stack. We say that $e_{i+1}$ and this stack element are \emph{tied to} this occurrence of $X$.
\end{definition}
The above means that there is a one-to-one correspondence between environments and occurrences of fixpoint variables: Reading a fixpoint variable $X$ in a configuration entails creation of a new environment, sometimes denoted by $e_X$, and every environment $e$ is created by an occurrence of a fixpoint variable $X_e$. Moreover, each priority on the priority stack is tied to a unique occurrence of a fixpoint and, hence, environment. The converse does not hold, since priorities can be removed from the priority stack. However, we will see below that environments that correspond to deleted priorities are not relevant to the remainder of a run.
\begin{definition}
\label{def:compute}
Let $C_i = (t_i,(\psi_i,e_i), e'_i, \Gamma_i, \Delta_i)$ be a configuration. The automaton is said to be currently \emph{computing the fixpoint $X$} if $e'_i$ was created by an occurrence of $X$. It is currently computing the environment $e'_i$, which is tied to an occurrence of $X$.
\end{definition}

\begin{lemma}
\label{lem:env-run-well}
Let $(C_i)_{i \in \mathbb{N}} = ((t_i,(\psi_i,e_i), e'_i, \Gamma_i,\Delta_i)_{i \in \mathbb{N}}$ be a run. For some $i$, let $C_i = (t_i,(\psi_i,e_i), e'_i, \Gamma_i, \Delta_i)$ be a configuration in that run. 
\begin{enumerate}
\item\label{lem:env-run-well-1} $e_i$ is either $e'_i$ or a predecessor of $e'_i$,
\item\label{lem:env-run-well-2} For $e'_i$, all variable bindings point to closures $(\psi,e')$ where $e'$ is a predecessor of $e'_i$ and the analogous property holds for all of $e'_i$s predecessors, 
\item\label{lem:env-run-well-3} all closures $(\psi,e')$ on the stack are such that $e'$ is either $e_i$ or a predecessor of $e_i$, 
\item\label{lem:env-run-well-4}  The sequence of priorities on the priority stack is exactly the sequence of priorities tied to  $e'_i$ and the sequence of its predecessors.
\end{enumerate}
\end{lemma}
\begin{proof}
The proof is by induction over the sequence of configurations. After the initial state is expanded to its transition relation, the lemma holds.  Item~\ref{lem:env-run-well-2} needs to be verified only on environment creation, Item~\ref{lem:env-run-well-1} only when the fixpoint currently being computed changes.

Now consider the form of the current closure $(\psi_i,e_i)$ and assume that the lemma holds so far. Clearly, for modal and boolean operators there is nothing to prove. 

If $\psi_i$ is of the form $(\psi'\, \psi'')$, 
then $\psi''$ is put on the stack and, by assumption, is either from $e_i$ or a predecessor environment, so again the new element conforms to Item~\ref{lem:env-run-well-3}. 

If $\psi_i$ is of the form $X$, then a new environment $e_{i+1}$ is created and will be the new environment currently being computed. Moreover, the parent environment of $e_{i+1}$ is $e_i$. This satisfies Item~\ref{lem:env-run-well-1}. Since 
all 
closures on the stack are from $e_i$ or from predecessors of $e_i$, the new environment satisfies Item~\ref{lem:env-run-well-2}. Since the stack is empty, it fulfills the stack requirements. Moreover, a new priority is added to the priority stack. Since it is tied to the new environment, Item~\ref{lem:env-run-well-4} continues to hold.

If $\psi_i$ is a variable not of ground type, $e_i$ switches to a predecessor and all items continue to hold.
If $\psi_i$ is  a variable of ground type, the stack is empty and, hence, Item~\ref{lem:env-run-well-3} is satisfied. There are two cases: Either $e_i(x) = (Q', e')$ with $e' = e'_i$, or, by Item~\ref{lem:env-run-well-2}, $e'$ is a predecessor of $e_i$ and, by Item~\ref{lem:env-run-well-1}, of $e'_i$. In the first case, the next closure will be $(Q',e')$ computed in $e'$, and there is nothing left to prove. In the second case, the automaton transitions towards $(t_i,(\psi_i,e_i), e''_i, \Gamma_i, \Delta'_i)$, where $e''_i$ is the parent of $e'_i$ and $ \Delta'_i$ is $\Delta_i$ with the top priority removed. Hence, Item~\ref{lem:env-run-well-4} is satisfied. Since $e'$ is a predecessor of $e'_i$, it is either equal to $e''_i$ or a predecessor of $e''_i$, so Item~\ref{lem:env-run-well-1} is also satisfied.
\end{proof}
From the definition of the transition relation, we can deduce that the environment which is currently being computed changes in two ways: By entering a new environment from its parent, which corresponds to environment creation, or by returning to the parent environment from an immediate successor environment. This means that, once an environment is left in favor of the parent environment, it will neve be returned to and the computation of its fixpoint is finished. Moreover, closures with this environment also never appear again. Hence, if such an environment is permanently left, we say that it is $\emph{being closed}$. Formally, a closed environment is one such that a variable of ground type from this environment has been read or, equivalently, the automaton has reached a configuration $(t,(Q,e), e', \Gamma, \Delta)$ such that $e'$ is the parent of the environment in question. Note that an environment is closed if and only if the corresponding priority has been removed from the priority stack. 

\begin{lemma}
\label{lem-decreasing-subexpr}
Let $e$ be an environment, let $(\psi,e)$ be a closure of ground type for some configuration and let $e$ be the environment currently being computed. 
 As long as $e$ stays the environment currently being computed, the type order of the current closure never properly decreases. 
 If the computation changes from $e$ to a proper successor and later returns to $e$ for the next time, this happens in a ground-type proper subexpression of $\psi$.
\end{lemma}
\begin{proof}
By the definition of the transition relation. The only transition that decreases the type order of the current closure is reading a fixpoint variable, which will change the environment currently being computed. Since $(\psi, e)$ is of ground type, the stack must be empty. If the computation leaves $e$ for a proper successor, this is through creation of a new environment or, equivalently, through reading a fixpoint variable. If the new environment binds a variable of ground type, the closure this variable points to must have been put on the stack between reading $(\psi,e)$ and the environment's creation. Hence, it must be a proper subexpression of $\psi$. If the new environment does not bind a variable of ground type, the computation can not return to $e$.
\end{proof}
 Since $\delta(X)$ for each $X$ has a finite syntax tree, repeated application of the previous lemma yields that the computation changes to any environment only a finite number of times. Otherwise we would obtain an infinitely descending sequence of subformulae of $\delta(X)$ where each subformula is an operand-type strict subformula of the previous.
 
It follows that each environment is either eventually closed, or eventually left permanently. Each environment appears as the environment currently being computed only finitely often. Moreover, each environment can only have finitely many direct successors because creation of a succesor of $e$ during a configuration requires the previous configuration to be in $e$. This means that, during an infinite run, infinitely many environments will not be closed and the corresponding priorities will never be popped from the priority stack. We define that \verifier wins the acceptance game if the highest priority occurs infinitely often but is never popped from the stack is even.

More formally, consider a run $(C_i)_{i \in \mathbb{N}}$. Consider the subsequence of configurations $(C'_j)_{j \in J}$ such that $C'_j = (t_{j}, (\delta(X),e_{j}), e_{j}, \epsilon,  \Delta_j)$, i.e., a configuration such that $e_j$ was created in this configuration, but such that there is no $i>j$ with a configuration $C_i = (t_{i}, (x,e_{j}), e_{j}, \Gamma_i,  \Delta_i)$ with $x$ of ground type, i.e., $e_j$ is never closed. By the above considerations, $J$ must be infinite. Then for all $n \geq j$, the priority stack $\Delta_j$ will be an initial segment of $\Delta_n$. In particular, this holds for all $n \in J$. Hence, the set $(\Delta_j)_{j \in J}$ is is such that $\Delta_j$ is a prefix of $\Delta_{j'}$ if $j \leq j'$. We define that a play is accepting if the highest priority that occurs in the limit of this prefix-ordered chain is even. We say that an automaton $\mathcal{A}$ accepts a tree $\mathcal{T}$, and write $\mathcal{T} \models \mathcal{A}$, if and only if \verifier has a strategy such that the acceptance game generates an accepting run. Note that the above constitutes a \emph{stair parity condition} in the sense that only those priorities contribute to the winning condition that are never removed from the priority stack. 

Note that this is not the same as just taking the sequence of priorities occurring during the run: It is possible that a high priority occurs infinitely often during the run, but each occurrence is eventually removed from the priority stack. This occurs in Example~\ref{example-inf-cont} where priority $1$ occurs infinitely often, but is always removed again from the priority stack a few configurations later.
\begin{definition}
Two \aapkm are  \emph{equivalent} if and only if they accept the same trees.
\end{definition}

\begin{observation}
For each \aapkm $\mathcal{A}$ there is an \aapkm $\overline{\mathcal{A}}$ over the same set of propositions such that for all trees, we have $\mathcal{T} \models \mathcal{A}$ if and only if $\mathcal{T} \not \models \overline{\mathcal{A}}$.
\end{observation}
The desired automaton is obtained by increasing the priorities of each state by one and replacing modal and boolean operators by their duals. A proof by induction over the structure of the acceptance game shows that a winning strategy for \verifier in the game for one automaton yields a winning strategy for \spoiler in the other, and vice versa.

\section{\apkm and \hlmu}
\label{sec:hfl}

\subsection{Syntax of \hlmu}

In addition to the set $\mathcal{P}$ of atomic propositions, fix infinite sets of variables $\mathcal{V}$ disjoint from $\mathcal{F}$ and $\mathcal{Y}$ disjoint from $\mathcal{X}$ that denote variables bound by a $\lambda$-expression, respectively a fixpoint quantifier. Separating $\mathcal{V}$ and $\mathcal{Y}$ is usually not done for \hlmu, but facilitates technical exposition. Lower case letters $x,y,\dotsc$ denote variables in $\mathcal{V}$, upper case letters $X,Y,\dotsc$ those in $\mathcal{Y}$.  

\hlmu-formulae $\varphi$ are defined by the grammar
\begin{align*}
   \varphi  & \Coloneqq    P \mid \neg P\mid \Diamond \varphi \mid \Box \varphi \mid \varphi \vee \varphi \mid \varphi \wedge \varphi \mid x \mid X \mid \lambda (x\colon\tau). \varphi \mid (\varphi\,\varphi)  \mid \mu (X\colon\tau).\varphi \mid \nu (X\colon\tau).\varphi 
\end{align*}
where $P \in \mathcal{P}$, $x \in \mathcal{V}$ and $X \in \mathcal{Y}$ and $\tau$ is a simple type. Note that negation is not present explicitly  in the logic since it can be eliminated \cite{DBLP:journals/corr/Lozes15}.

The binder $\lambda (x^v\colon\tau).\varphi$ binds $x$ in $\varphi$, the binder $\sigma (X\colon\tau).\varphi$ with $\sigma \in  \{\mu, \nu\}$ binds $X$ in $\varphi$. Let $\operatorname{sub}(\varphi)$ be the set of subformulae of $\varphi$. An \hlmu-formula is \emph{well-named} if there is, for each $X \in \mathcal{Y}$, at most one subformula of the form $\sigma (X\colon\tau).\psi$ and, for each $x \in \mathcal{V}$, at most one subformula of the form $\lambda (x\colon\tau).\psi$.

A variable from $\mathcal{V}$ or $\mathcal{Y}$ in a formula $\varphi$ is \emph{bound} if it is bound by a binder of the respective type, and \emph{free} otherwise. A formula is called \emph{closed} if it has no free variables and \emph{open} otherwise.
For a well-named formula $\varphi$ and $X \in \mathcal{Y} \cap \operatorname{sub}(\varphi)$, define $\operatorname{fp_\varphi}(X)$ as the unique subformula $\psi$ of $\varphi$ such that $\psi = \sigma X.\psi'$ for $\sigma  \in \{\mu,\nu\}$. We have a partial order $<_{\operatorname{fp_\varphi}}$ on the fixpoint variables of  $\varphi$ via $X <_{\operatorname{fp_\varphi}} Y$ if $Y$ appears freely in $\operatorname{fp_\varphi}(X)$. We say that $Y$ is \emph{outermore} than $X$. A variable is outermost among a set of variables if it is maximal in this set with respect to $<_{\operatorname{fp_\varphi}}$.

We say that $\varphi$ has type $\tau$ in a \emph{context} $\Sigma$ if $\Sigma \vdash \varphi \colon \tau$ can be derived via the typing rules in in Figures~\ref{aapka-typing-rules} and \ref{hlmu-typing-rules}. Note that the rules concerning variables from $\mathcal{X}$ and $\mathcal{F}$ are not used. If $\Sigma \vdash \varphi \colon \tau$ for some $\Sigma$ and $\tau$ then $\varphi$ is \emph{well-typed}. A closed formula is well typed if $\emptyset\vdash \varphi\colon\tau$.  Typing judgments are unique if formulae are annotated with the correct types \cite{viswaviswa}. We usually omit the type annotations and tacitly assume that all formulae are well-typed and that the type of a formula can be derived from context.

\begin{figure}
\caption{Additional Typing Rules for \hlmu.}
\label{hlmu-typing-rules}
\begin{mathpar}
\inferrule{ }{\Sigma, x \colon\tau \vdash x\colon\tau} \and
\inferrule{ }{\Sigma, X \colon\tau \vdash X\colon\tau} \and
\inferrule{\Sigma,x\colon\tau\vdash \varphi\colon\tau'}{\Sigma\vdash \lambda (x \colon \tau).\varphi\colon\tau\rightarrow\tau'} \and
\inferrule{\Sigma,X \colon \tau \vdash \varphi\colon\tau}{\Sigma \vdash \sigma (X
\colon\tau). \varphi\colon\tau} 
\end{mathpar}
\end{figure}

\subsection{Semantics of \hlmu}

Fix a tree $\mathcal{T}$. The semantics of types are partially ordered sets defined inductively via
 $\llbracket\Prop\rrbracket = (2^T,\subseteq)$ and
	 $\llbracket \tau \rightarrow \tau' \rrbracket = (\llbracket\tau'\rrbracket^{(\llbracket\tau\rrbracket)},\sqsubseteq_{\tau \rightarrow\tau'}),$
where $T =  \{0,1\}^*$ and $\llbracket\tau'\rrbracket^{(\llbracket\tau\rrbracket)}$ is the set of monotone  functions from $\llbracket\tau\rrbracket$ to $\llbracket\tau'\rrbracket$. Define the partial order $\sqsubseteq_{\tau \rightarrow\tau'}$ via pointwise comparison: For $f,g \in \llbracket\tau'\rrbracket^{(\llbracket\tau\rrbracket)}$ let $f \sqsubseteq_{\tau \rightarrow\tau'} g$ if and only if $f(x) \sqsubseteq_{\tau'} g(x)$ for all $x \in {\llbracket\tau\rrbracket}$.

Note that $\llbracket\Prop\rrbracket$ is a boolean algebra and, hence, also a complete lattice. This makes $\llbracket \tau \rightarrow \tau' \rrbracket$ also a complete lattice for all $\tau,\tau'$. Let $\bigsqcup_{\tau} M$ and $\bigsqcap_{\tau} M$ denote the join and meet, respectively, of the set $M \subseteq \llbracket\tau\rrbracket$, and let $\top_{\tau}$ and $\bot_\tau$ denote the maximal and minimal elements of  $\llbracket\tau\rrbracket$. 
%

Let $\Sigma = X_1 \colon\tau_1, \dotsc, X_n\colon\tau_n, x_1\colon \tau'_1,\dotsc,x_m\colon\tau'_m$ be a context. An interpretation $\eta$ is a partial map from the sets of variables $\mathcal{V}$ and $\mathcal{Y}$ such that $\eta(X_i) \in \llbracket\tau_i\rrbracket$ for all $i \leq n$ and $\eta(x_j) \in \llbracket\tau'_j\rrbracket$ for all $j\leq m$. Then $\eta[X \mapsto f]$ is the interpretation that maps $X$ to $f$ and agrees with $\eta$ otherwise, similar for $\eta[x \rightarrow f]$.

We define the semantics of \hlmu\ over \Transsys inductively as in Figure~\ref{hlmu-semantics} (with dual cases left out for space considerations).
\begin{figure}
\caption{Semantics of \hlmu.}
\label{hlmu-semantics}
\begin{align*}
	 \llbracket \Sigma\vdash P\colon\Prop\rrbracket_\eta                        & = P^\Transsys \\
	 \llbracket \Sigma\vdash\varphi \vee \psi\colon\Prop\rrbracket_\eta   & = \llbracket\Sigma\vdash\varphi\colon\Prop \rrbracket_\eta\cup \llbracket\Sigma\vdash\psi \colon \Prop \rrbracket_\eta\\
	\llbracket \Sigma\vdash \Diamond \varphi \colon\Prop\rrbracket_\eta &  = \big\{t \in T\colon\ t0  \in \llbracket\Sigma\vdash\varphi\colon\Prop\rrbracket_\eta \text{ or } t1  \in \llbracket\Sigma\vdash\varphi\colon\Prop\rrbracket_\eta  \big\}\\
	\begin{split} \llbracket \Sigma\vdash \lambda (x\colon\tau) . \varphi: \tau \rightarrow \tau' \rrbracket_\eta  & =  \{ f \in \llbracket \tau \to \tau'\rrbracket \colon  \forall y\in\llbracket\tau\rrbracket. \\ & \qquad \qquad f(y) =  \llbracket\Sigma,x\colon\tau \vdash \varphi:\tau' \rrbracket_{\eta[x\mapsto y]}\}\end{split}\\
	\llbracket \Sigma\vdash X\colon\tau\rrbracket_\eta & = \eta(X) \\
	\llbracket \Sigma\vdash x\colon\tau\rrbracket_\eta & = \eta(x)\\
	\llbracket \Sigma\vdash \mu (X\colon\tau) .\varphi\colon\tau\rrbracket_\eta & = \bigsqcap \big\{ d \in  \llbracket\tau\rrbracket \colon \llbracket\Sigma,(X\colon\tau) \vdash \varphi\colon\tau\rrbracket_{\eta[X \mapsto d]} \sqsubseteq_\tau d\big\}\\
	\llbracket \Sigma\vdash (\varphi\,\psi)\colon\tau'\rrbracket_\eta & = \llbracket\Sigma \vdash \varphi\colon\tau\rightarrow \tau'\rrbracket_\eta (\llbracket\Sigma\vdash\psi\colon\tau\rrbracket_\eta) 
\end{align*}
\end{figure}
For well-typed, well-named $\varphi$, we write $\mathcal{T},t\models_\eta \varphi$ if $s\in \llbracket\emptyset~\vdash~\varphi:\Prop\rrbracket_\eta$. We write $\mathcal{T},t \models \varphi$ if $\varphi$ is closed and $\eta$ is the empty interpretation.
Two formulae are \emph{equivalent}, written $\varphi \equiv \varphi'$, if  $\llbracket \Sigma\vdash \varphi \rrbracket_\eta = \llbracket \Sigma \vdash \varphi'\rrbracket_\eta$ for all $\eta$, $\Sigma$.

\subsection{Translations between \hlmu and \apkm}

\begin{lemma}
\label{lem:embedding-apkm}
Let $\varphi$ be an \hlmu-formula of order at most $k$. Then there is an \apkm $\mathcal{A}_\varphi$ of order at most $k$ such that, for all trees $\mathcal{T},t$, we have $\mathcal{T},t \models \varphi$ if and only if $\mathcal{A}_\varphi$ accepts $\mathcal{T},t$.
\end{lemma}
\begin{proof}{(Sketch)}
For space considerations, we only give a sketch of the proof. Let $\varphi$ be a \hlmu-formula. 

Since lambda abstraction is implicit for \apkm and can only occur directly after a fixpoint, occurrences of lambda abstraction $\lambda f.\psi$ in $\varphi$ that are not of the form $\sigma X.\lambda f_1.\dotsc \lambda f_n.\psi$ need to be padded by vacuous fixpoints. If $f$ is of type $\tau_1$ and $\psi$ is of type $\tau_2$, replace $\lambda f.\psi$ by $\sigma X.\lambda f.\psi$, where $X$ is of type $\tau_1\to\tau_2$ and $\sigma$ is chosen as convenient. 

Next, free lambda variables are removed. For a subformula $\sigma X.\psi$ that contains a free variable $f$ that is not a fixpoint, replace $\sigma X.\psi$ by $((\sigma X. \lambda f'. \psi[f'/f])\,f)$ where $f'$ is of the same type as $f$. This is organized such that fixpoints are translated before fixpoints in their subformulae, i.e., from top to bottom.

In a third step, any fixpoint of the form $\sigma X. \lambda f_1.\dotsc \lambda f_n.\psi$ with $\psi$ of type $\tau_1\to\dotsb\to\tau_m\to\Prop$ is changed to its $\eta$-long form, i.e., to $\sigma X.\lambda f_1.\dotsc\lambda f_n.\lambda g_1.\dotsc \lambda g_m. \psi'$ with $\psi' =  ((\psi\,\tau_m)\dotsb)\,\tau_1)$. 

It is not hard to verify that neither of these steps changes semantics of the formula in question. Let $\varphi'$ be the resulting \hlmu-formula and let $\mathcal{X}$ be the collection of fixpoint variables in $\varphi'$. Without loss of generality, $\varphi'$ has the form $\sigma X_{\text{init}}.\varphi''$ for some $\sigma$. 

For each fixpoint $X$ with defining formula $\sigma X.\lambda f_1.\dotsc\lambda f_n.\psi$ set $\delta(X)$ to $\psi$ where all occurrences of formula of the form $\sigma' X'.\psi'$ are replaced by $X'$ and set $\tau_X$ as the type of $X$. Then the automaton
$
\mathcal{A} = (\mathcal{X}, \Lambda, X_{\text{init}}, \delta, (\tau_X)_{X \in \mathcal{X}})
$
with $\Lambda$ chosen such that each fixpoint is labeled odd or even depending on parity, but not lower than any fixpoint in a subformula, is an \apkm accepting the same trees as $\varphi$. 
\end{proof}

\begin{lemma}
\label{lem:embedding-hfl}
Let $\mathcal{A}$ be an \aapkm of order at most $k$. Then there is an \hlmu-formula $\varphi_{\mathcal{A}}$ of order at most $k$ such that, for all trees $\mathcal{T},t$, we have $\mathcal{T},t \models \varphi_{\mathcal{A}}$ if and only if $\mathcal{A}$ accepts $\mathcal{T},t$.
\end{lemma}
We skip the proof for space considerations. It rests on the idea that a fixpoint state $X$ computes the formula $\sigma X. \lambda f_1^{X}.\dotsc \lambda f_{n_X}^X. \delta(X)$ where $\sigma $ is $\mu$ if $\Lambda(X)$ is odd and $\nu$ otherwise. However, the translation is subject to the same exponential blowup in size (but not in order) that occurs when translating ordinary parity automata into the modal $\mu$-calculus. Moreover, further preprocessing is necessary because fixpoints can occur as operator-operand pair where the operator has a higher priority. In this case, a duplication of arguments is necessary to ensure proper precedence of fixpoints in the syntax tree. \footnote{This idea is due to Naoki Kobayashi and {\'{E}}tienne Lozes.}


\begin{corollary}
Emptiness of \aapkm is undecidable.
\end{corollary}

\begin{corollary}
For any \emph{finite} tree $\Transsys$ (or any finite Kripke structure), and any \aapkm $\mathcal{A}$ of order $k$, it is decidable in time $k$-fold exponential in the size of $\mathcal{T}$ whether $\mathcal{T}\models\mathcal{A}$.
\end{corollary}

\section{The Alternation Hierarchy for Alternating Parity Krivine Automata}

\subsection{Alternation Classes}

We define the 
semantic alternation class via the least number of priorities of any equivalent automaton.

\begin{definition}
We define the classes
\begin{itemize}
\item $\Sigma^{\operatorname{sem}}_n$ as the set of all \aapkm equivalent to one with at most $n$ priorities such that the highest is even
\item $\Pi^{\operatorname{sem}}_n$ as the set of all \aapkm equivalent to one with at most $n$ priorities such that the highest is odd
\item $\Delta^{\operatorname{sem}}_n$ = $\Sigma^{\operatorname{sem}}_n \cap \Pi^{\operatorname{sem}}_n$.
\end{itemize}
\end{definition}

\begin{remark}
The following inclusions hold:
\begin{align*}
  \Sigma^{\operatorname{sem}}_n &\subseteq \Sigma^{\operatorname{sem}}_{n+1} &   \Pi^{\operatorname{sem}}_n &\subseteq \Pi^{\operatorname{sem}}_{n+1} \\
  \Sigma^{\operatorname{sem}}_n &\subseteq \Pi^{\operatorname{sem}}_{n+1} &  \Pi^{\operatorname{sem}}_n &\subseteq \Sigma^{\operatorname{sem}}_{n+1} \\
  \Sigma^{\operatorname{sem}}_n &\subseteq \Delta^{\operatorname{sem}}_{n+1} & \Pi^{\operatorname{sem}}_n &\subseteq \Delta^{\operatorname{sem}}_{n+1} \\
  \Delta^{\operatorname{sem}}_n & \subseteq \Delta^{\operatorname{sem}}_{n+1}& &
\end{align*}
\end{remark}

Note that the alternation classes are independent of the order of an automaton. For a \hlmu-formula $\varphi$, we say that $\varphi$ is in some alternation class if there is an equivalent \apkm in that class.

\begin{observation}
If $\mathcal{A} \in \Sigma^{\operatorname{sem}}_n$ then $\overline{\mathcal{A}} \in \Pi^{\operatorname{sem}}_n$, if $\mathcal{A'} \in \Pi^{\operatorname{sem}}_n$ then $\overline{\mathcal{A}'} \in \Pi^{\operatorname{sem}}_n$.
\end{observation}


\subsection{Trees Encoding Acceptance Games}
For each $n \geq 1$, define a set of propositions $\mathcal{P}_n$ as $\{D,C,V,T,F,F_0,\dotsc,F_{n-1}\}$ as well as a set $\mathcal{P}'_n$ as $\{D,C,V,T,F,F_1,\dotsc,F_{n}\}$. 

Let $n \geq 1$. Consider a  tree $\mathcal{T}$ and some \aapkm $\mathcal{A}$ with at most $n$ priorities, over $\mathcal{P}_n$ or $\mathcal{P}'_n$ (depending on whether the highest priority is odd or even). We construct a tree $T(\mathcal{T},\mathcal{A})$ over the same set of propositions which encodes the game tree  of the acceptance game $G(\mathcal{T},\mathcal{A})$ of $\mathcal{A}$ over $\mathcal{T}$. A state labeled by $C$ signals that \spoiler picks a successor configuration, a state labeled by $D$ signals that \verifier picks a successor configuration, a state labeled by $F_i$ signals that priority $i$ is added to the priority stack in this configuration and a state labeled by $V$ signals that the top priority is being removed. Configurations where the priority stack is not being manipulated and neither player picks a successor configuration are treated as if \verifier picks a successor, but both subtrees of $T(\mathcal{T},\mathcal{A})$ are isomorphic.

The tree is generated inductively. Each position $(t,(Q,e), e', \Gamma, \Delta)$ in the acceptance game induces a subtree, with the root of the tree being generated by the initial position. At each vertex, exactly one proposition $P$ from $\mathcal{P}_n$, respectively $\mathcal{P}'_n$ is true. We say that this vertex is labeled by $P$.
\begin{itemize}
\item The subtree induced by a position with $Q$ of the form $X$ is labeled $F_{\Lambda(X)}$. Both children are the subtree induced by $(t,(\delta(X), e''), e'', \Gamma',\Delta')$ where $e'', \Gamma'$ and $\Delta'$ are as per the transition relation.

\item The subtree induced by a position with $Q$ of the form $f^X_{j}$ is labeled $V$ if both $f^X_j$ is of type $\Prop$ and $e(f^X_{j}) = (Q', e'')$ with $e'\not = e''$. Otherwise, it is labeled $D$. Both children are the subtree induced by the successor configuration as per the transition relation.

\item The subtree induced by a position with $Q$ of the form $(\psi_1 \psi_2)$ is labeled $D$. Both children are the subtree induced by $(t,(\psi_1,e),e', \Gamma (\psi_2,e),\Delta)$.


\item The subtree induced by a position with $Q$ of the form $\psi_1 \vee \psi_2$ is labeled $D$. The left subtree is the subtree induced by $(t,(\psi_1,e),e',\Gamma,\Delta)$, the right subtree is that induced by $(t,(\psi_2,e),e',\Gamma,\Delta)$.

\item The subtree induced by a position with $Q$ of the form $\psi_1 \wedge \psi_2$ is labeled $C$. The left subtree is the subtree induced by $(t,(\psi_1,e),e',\Gamma,\Delta)$, the right subtree is that induced by $(t,(\psi_2,e),e',\Gamma,\Delta)$.

\item The subtree induced by a position with $Q$ of the form $\Diamond \varphi$ is labeled $D$. The left subtree is the subtree induced by $(t0,(\varphi,e),e',\Gamma, \Delta)$, the right subtree is that induced by $(t1,(\varphi,e),e',\Gamma, \Delta)$. 

\item The subtree induced by a position with $Q$ of the form $\Box \varphi$ is labeled $C$. The left subtree is the subtree induced by $(t0,(\varphi,e),e',\Gamma, \Delta)$, the right subtree is that induced by $(t1,(\varphi,e),e',\Gamma, \Delta)$. 

\item The subtree induced by a position with $Q$ of the form $P$ or $\neg P$ is labeled $T$ if $\mathcal{T},s \models Q$ and $F$ else. Both children are the subtree induced by $(t,(Q,e),e', \Gamma, \Delta)$ again.
\end{itemize}

It is easy to verify that this defines an infinite, fully binary tree. Figure~\ref{runtree2} shows an example. 
\begin{figure}
\caption{Part of a $T(\mathcal{A}, \mathcal{T})$ for $\mathcal{A}$ and $\mathcal{T}$ from Example~\ref{example-inf-cont}. Omitted subtrees are isomorphic to their sibling if present or not shown for space considerations. $C_i$ refers to the configuration from Figure~\ref{run-exml} that induces the subtree and is not part of the label.}
\label{runtree2}
\begin{tikzpicture}
[level distance=0.7cm ,sibling distance=1cm,
  level 4/.style={sibling distance=3cm},
   level 5/.style={sibling distance=1.5cm},
  level 8/.style={sibling distance=2cm}]
  \node {$C_0,F_1$}
    child {node {\dots}}
    child {node {$C_1, D$}
      child {node {\dots}}      
      child {node {$C_2, F_1$}
        child {node {\dots}}
        child { node {$C_3, D$}
                   child { node {$D$}
                   child {node {\dots}}        
                     child { node {$V$}
                     child {node {\dots}}        
                       child { node {$D$}
                       child {node {\dots}}        
                         child { node {$F$}
                           child {node {$F$}
                             child {node {\dots}  }  }
                           child {node {\dots}  }    
              }
              }
              }                                      
                   }
                child { node {$C_4, C$}
                child {node {\dots}}        
                  child { node {$C_5, F_0$}
                  child {node {\dots}}        
                    child { node {$C_6, D$}
                    child {node {\dots}}  
                      child { node {$C_7, F_1$}
                      child {node {\dots}}        
                        child { node {$C_8, D$}
                        child {node {$C$}
                            child {node {\dots}}
                           }        
                          child { node {$C_9, D$}
                          child {node {\dots}}        
                             child { node {$C_10, V$}
                             child {node {\dots}}        
                               child { node {$C_{11}, D$}
              }
              }
              }
              }
              }      
              }
              }
              }        
              }
        }
    };
\end{tikzpicture}
\end{figure}
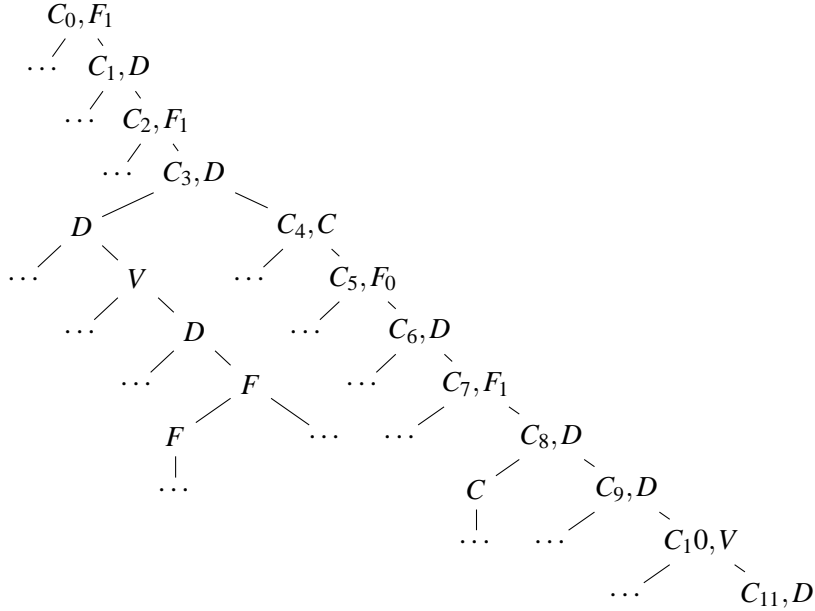

\subsection{Hard Automata}

We now construct \aapkm that are hard for their alternation classes. Following Arnold's \cite{arnold} and Lange's \cite{infcomp06} proofs, these automata accept trees enconding an acceptance game that is won by \verifier, respectively \spoiler.

Consider the $\mathcal{P}_n$-, respectively $\mathcal{P}'_n$-\aapkm $\mathcal{A}^\Sigma_n$  and $\mathcal{A}^\Pi_n$ defined for each $n \geq 1$ as follows:
\begin{itemize}
\item The fixpoint states are $\{I,O, X_{n-1},\dotsc,X_0\}$, 
\item the type of $I$ is $\Prop$, the type of the other states is $\Prop \to\Prop$,
\item the initial state is $I$,
\item $\Lambda(I) =\Lambda(O) = 0, \Lambda(X_i) = i$ for $\mathcal{A}^\Sigma_n$ and $\Lambda(I) = \Lambda(O) = 1$ and $\Lambda(X_i) = i+1 $ for $\mathcal{A}^\Pi_n$,
\item $\delta(I) = O\,\true$, $\delta(X_i) = x_0\colon\Prop \mapsto (X_{i-1}\, x_0)$ for $i>0, \delta(X_0) = x_0 \mapsto (O\,x_0)$,
\item $\delta(O) = x_0 \mapsto \neg F \wedge ( T \vee \bigwedge 
\left(\begin{aligned}
(D) & \rightarrow \Diamond(O\, x_0) \\ 
(C) & \rightarrow \Box (O\, x_0) \\    
(V) & \rightarrow \Diamond (x_0) \\ 
(F_{n-1}) & \rightarrow \Diamond \big(X_{n-1}\, (O\, x_0)\big) \\ 
          & \dots \\
(F_{0}) & \rightarrow \Diamond \big(X_{0}\, (O\, x_0)\big) \\ 
\end{aligned}\right) ).$
\end{itemize}
Again, it is easy to verify that $\mathcal{A}^\Sigma_n \in \Sigma^{\operatorname{sem}}_n$ and that $\mathcal{A}^\Pi_n \in \Pi^{\operatorname{sem}}_n$.

These automata are equivalent to the \hlmu-formulae $(\sigma_{n-1} X_{n-1}. \lambda x_0 (\big(\dotsb \sigma X_0. \lambda x_0. \psi\big)\dotsb) \, x_0\true$ where the $\sigma_i$ are $\mu$, respectively $\nu$ depending on the alternation class, and $\psi = \delta(O)$.
\begin{definition}
\label{def-rounds}
Consider a play of $\mathcal{A}^\Sigma_n$ over a $\mathcal{P}_n$-tree, respectively of $\mathcal{A}^\Pi_n$ over a $\mathcal{P}'_n$-tree generated from an acceptance game. A \emph{round} in this play consists of a configuration where the current closure is $O$ and all subsequent configurations until it is $O$ again. An environment is \emph{tied to a round} if it is created during that round.
\end{definition}
A round begins with the automaton in $O$. Unless the current tree node is labeled by $F$, \spoiler chooses the right conjunct in $\delta(O)$, and \verifier, unless the current state is labeled by $T$, chooses the right disjunct. \spoiler then picks the conjunct indicated by the label of the current subtree in the big conjunction and \verifier picks the right part of the implication. Any different choice results in the player making that choice instantly loosing the game. One of the players is then in charge of picking a successor subtree. Depending on the conjunct picked by \spoiler, the game continues in a new instance of $O$, goes through $X_i,\dotsc,X_0,O$ for some $i$ or continues with the content of $x_0$. The latter will always lead to another instance of $O$, as we will see below. In either case, the game continues in the next round.
\begin{observation}
Each round corresponds to exactly one configuration in the acceptance game of $\mathcal{A}$, namely that which induces the subtree in $T(\mathcal{T},\mathcal{A})$ during the first configuration of the round. Furthermore, the current subtree in the game for $\mathcal{A}^\Sigma_n$ is labeled by $C$ if and only if the configuration that induces it has a conjunction or a box as the top operator in the formula part of the current closure. 
\end{observation}
Note that each configuration in a play for $\mathcal{A}^\Sigma_n$, respectively $\mathcal{A}^\Pi_n$ over suitable trees is part of exactly one round, with the exception of the first two configurations which have current closures $(I)$ and $(O\, \true)$.

We call a round a \emph{$V$-round} if \spoiler picks the conjunct with $V$ on the left of the implication, we call it an \emph{$F_k$-round} if he picks the conjunct with $F_k$ on the left of the implication and we call a round a \emph{plain round} if he picks the conjuncts with $C$ or $D$ on the left of the implication. A round is \emph{closed} if the environment tied to the single occurrence of $O$ during that round is closed. A $V$-round is always closed immediately.

\begin{lemma}
\label{stack-prio}
Consider a play of $\mathcal{A}^\Sigma_n$, respectively $\mathcal{A}^\Pi_n$ over a $\mathcal{P}_n$-tree, respectively $\mathcal{P}'_n$-tree generated from an acceptance game and let the automaton be at the 
start of some round, i.e., just before reading another occurrence of $O$. Let $(R_i)_{i \in I}$ be the sequence of unclosed rounds played so far, in order. Set $p^\Sigma(R) = 0$ if $R$ is a plain round, set $p^\Pi(R) = 1$ if $R$ is a plain round, set $p^\Sigma(R) = 0,k,\dotsc,1,0$ if $R$ is an $F_k$-round and set $p^\Pi(R) = 1,k+1,\dotsc,2,1$ if $R$ is a plain round. Then the priority stack of $\mathcal{A}^\Sigma_n$ from bottom to top is the concatenation of the $p^\Sigma(R_i)$ from first to last and the priority stack of $\mathcal{A}^\Pi_n$ from bottom to top is the concatenation of the $p^\Pi(R_i)$ from first to last.

Moreover, all unclosed environments are tied to unclosed rounds. Tied to any plain round is a single environment for its ocurrence of $O$ and it binds $x_0$ the last environment of the first unclosed round before. Tied to an $F_k$-round is a sequence of environments for the occurrences of $O,X_k,\dotsc,X_0$. The environment for $X_0$ is the last environment, they all bind $x_0$ to $x_0$ of the previous environment except the environment for $X_k$ which binds $x_0$ to $(O\,x_0)$ in the environment for $O$ of its own round. Here, the initial unfolding for $I$ is considered a dummy round.
\end{lemma}
\begin{proof}
The proof is by induction over the play. At the beginning of the very first round, the priority stack contains only the priority for $I$ and $(\true,e_0)$ is on the stack. A plain round will consume the content from the stack, which is $(x_0,e)$ of the previous round, or $(\true,e_0)$ for the very first round, and tie $x_0$ of its single ocurrence of $O$ to it. Moreover, it will add $0$ to the priority stack. An $F_k$-round $R$ will also consume $(x_0,e)$, respectively $(\true,e_0)$ from the stack and tie $x_0$ of the single ocurrence of $O$ to it. During the round, the automaton will unfold $X_k$ and tie $(O\,x_0)$ of that environment to $X_k$'s $x_0$, then unfold $X_{k-1},\dotsc,X_0$ and create a chain of $x_0$ pointing to $x_0$ of the environment before. Moreover, it will put the sequence $p^\Sigma(R)$, respectively $p^\Pi(R)$ on the priority stack.

A $V$-round will put priority $0$ on the stack, tie the $x_0$ of its single occurrence of $O$ to $x_0$ of the previous unclosed round and then immediately read it. Consequently, all the environments of the previous unclosed round will be closed, including the ocurrrence of $0$, and all the priorites tied to it will be popped. Notably, this will close all unclosed previous plain rounds until the next $F_i$-round, but nothing more.
\end{proof}

\begin{lemma}
For all $\mathcal{P}_n$, respectively $\mathcal{P}'_n$-automata $\mathcal{A} \in \Sigma^{\operatorname{sem}}_n$ and all infinite, fully binary $\mathcal{P}_n$ -trees $\mathcal{T}$, we have that $\mathcal{T} \models \mathcal{A}$ if and only if $T(\mathcal{T},\mathcal{A}) \models \mathcal{A}^\Sigma_n$ and for all $\mathcal{P}_n$, respectively $\mathcal{P}'_n$-automata $\mathcal{A'} \in \Pi^{\operatorname{sem}}_n$, we have that $\mathcal{T} \models \mathcal{A'}$ if and only if $T(\mathcal{T},\mathcal{A'}) \models \mathcal{A}^\Pi_n$.
\end{lemma}
\begin{proof}
We only show the case for $\mathcal{P}_n$ and we only show that \verifier has a winning strategy in the acceptance game for $\mathcal{A}^\Sigma_n$ over $T(\mathcal{T},\mathcal{A})$ if she has one for $\mathcal{A}$ over $\mathcal{T}$, for a $\mathcal{P}_n$-automaton $\mathcal{A}$ in $\Sigma^{\operatorname{sem}}_n$ and $\mathcal{T}$ a $\mathcal{P}_n$-tree. The other cases are similar.
Assume that \verifier has a winning strategy in the latter game. 

The correspondence between rounds in the game for $T(\mathcal{T},\mathcal{A})$ and configurations in the game for $\mathcal{A}$ suggests the following strategy for \verifier in the former game: Stay within subtrees that represent configurations that follow her winning strategy. Since the underlying game is assumed to be winning for \verifier and the root of $T(\mathcal{T},\mathcal{A})$ represents such a configuration by assumption, she can maintain this invariant in any round where she picks the successor configuration. In rounds where \spoiler picks the successor configuration, both of his choices must be winning for \verifier in the underlying game for $\mathcal{A}$ over $\mathcal{T}$ for otherwise the current configuration would not be winning for \verifier. Clearly, following this strategy will guarantee that \verifier wins any finite play of the game for $\mathcal{A}^\Sigma_n$ by avoiding a node labeled $F$.

It remains to show that \verifier wins any infinite play when following the strategy above. This is because the sequence of unclosed nonplain rounds in the game for $\mathcal{A}^\Sigma_n$, and priority stack in the game for $\mathcal{A}$ correspond like this: If $(F_{k_i})_{i \in I}$ is the sequence of unclosed nonplain rounds, the $(k_i)_{i \in I}$ is the priority stack of $\mathcal{A}$. This follows from an induction over the two plays: Before the first round of the game for $T(\mathcal{T},\mathcal{A})$, the sequence of unclosed rounds is empty, and so is the priority stack of the correspondig configuration of $\mathcal{A}$. Any plain round will add a $0$, the least priority, to the priority stack of $\mathcal{A}^\Sigma_n$ and will not change the priority stack of $\mathcal{A}$. An $F_k$-round will add $k$ to the priority stack of $\mathcal{A}$ and will add an unclosed $F_k$-round to the play of $\mathcal{A}^\Sigma_n$. A $V$-round will remove one priority $k$ from the priority stack for $\mathcal{A}$ and will close a number of plain rounds and exactly one nonplain round in the game for $\mathcal{A}^\Sigma_n$. By the induction hypothesis, this is an $F_k$-round. 

Hence, after both plays are finished, the highest priority to occur infinitely often on the stack for $\mathcal{A}$ is $k$ if and only if there are infinitely many unclosed $F_k$-rounds, but only finitely many unclosed $F_{k'}$ rounds for $k'>k$. It follows from Lemma~\ref{stack-prio} that the highest priority to occur infinitely often on the stack for $\mathcal{A}^\Sigma_n$ is $k$ as well. Since \verifier wins the first game by assumption, that number must be even.
\end{proof}

\begin{lemma}
\label{lem:banachfpt}
For each $n \geq 1$ and every $\mathcal{P}_n$-\aapkm $\mathcal{A} \in \Sigma^{\operatorname{sem}}_n$, there is a unique $\mathcal{T}^*$ such that $T(\mathcal{T}^*,\mathcal{A}) = \mathcal{T}^*$. For each $n \geq 1$ and every $\mathcal{P}_n$ \aapkm $\mathcal{A} \in \Pi^{\operatorname{sem}}_n$, there is a unique $\mathcal{T}^*$ such that $T(\mathcal{T}^*,\mathcal{A}) = \mathcal{T}^*$.
\end{lemma}
\begin{proof}
The sets of all $\mathcal{P}_n$-trees, respectively the sets of all $\mathcal{P}'_n$-trees, form metric spaces via the metric $d(t,t') = 2^{-i}$, where $i$ is the first level on which $t$ and $t'$ differ. It is well known that these spaces are complete \cite{DBLP:journals/fuin/ArnoldN80}. Moreover, on all of these spaces, the mapping $f\colon \mathcal{T}\mapsto T(\mathcal{T}, \mathcal{A}^{\operatorname{sem}}_n)$ is a contraction in the sense of the Banach Fixpoint Theorem since the game trees of two trees that differ at level $i$ will coincide at least up to level $i+1$. This is because the game with $\mathcal{A}^{\operatorname{sem}}_n$ transitions through $I$ first and a full rotation through $\delta(O)$ for each level.
Hence, by the Banach Fixpoint Theorem,  $f$ has a fixpoint $\mathcal{T}^*$. 
\end{proof}

\begin{theorem}
$\mathcal{A}^\Sigma_n \in \Sigma^{\operatorname{sem}}_n \setminus \Pi^{\operatorname{sem}}_n$ and $\mathcal{A}^\Pi_n \in \Pi^{\operatorname{sem}}_n \setminus \Sigma^{\operatorname{sem}}_n$.
\end{theorem}
\begin{proof}
For the sake of contradiction, assume that $\mathcal{A}^\Sigma_n \in \Pi^{\operatorname{sem}}_n$. Then $\mathcal{A}' = \overline{\mathcal{A}^\Sigma_n} \in \Sigma^{\operatorname{sem}}_n$. By Lemma~\ref{lem:banachfpt}, there is $\mathcal{T}^*$ such that $T(\mathcal{T}^*,\mathcal{A}') = \mathcal{T}^*$. So by construction of $T(\mathcal{T}^*,\mathcal{A}')$, we have $\mathcal{T}^* \models \mathcal{A}^\Sigma_n$ iff $\mathcal{T}^* \models \mathcal{A}'$. But $\mathcal{A}' = \overline{\mathcal{A}^\Sigma_n}$, which is a contradiction. hence, $\mathcal{A}^\Sigma_n \notin \Pi^{\operatorname{sem}}_n$.

A similar proof works for the dual case.
\end{proof}

\begin{corollary}
For each $n$, $\Sigma^{\operatorname{sem}}_n \subsetneq \Sigma^{\operatorname{sem}}_{n+1}$ and $\Pi^{\operatorname{sem}}_n \subsetneq \Pi^{\operatorname{sem}}_{n+1}$.
\end{corollary}
\begin{proof}
Since $\Pi^{\operatorname{sem}}_n \subseteq \Sigma^{\operatorname{sem}}_{n+1}$, non-strictness of $\Sigma^{\operatorname{sem}}_n \subseteq \Sigma^{\operatorname{sem}}_{n+1}$ would contradict the previous theorem. The same argument works for the dual case.
\end{proof}

\section{Discussion}

It is a priori quite surprising that the order of an \aapkm or a \hlmu-formula is not of relevance when it comes to its alternation class. In particular, the automata $\mathcal{A}^\Sigma_n$ and $\mathcal{A}^\Pi_n$ that serve as example of automata that are hard for their respective classes are of order $1$. This is surprising, since for the \hlmu-model-checking problem, which corresponds to acceptance for \aapkm, complexity is almost exclusively dictated by the order of a formula. We believe that this dichotomy stems from the way the transition relation for \aapkm is restricted to formulae of ground type. A state that would compute a higher-order function, say of type $(\Prop\to \Prop) \to (\Prop \to \Prop)$ actually does not compute the full higher-order function, but its equivalent of type $(\Prop \to \Prop) \to \Prop \to \Prop$ at a fixed argument of type $\Prop$. The first case requires computations over the full extent of a higher-order lattice, while in the second case it is sufficient to find an approximation that is good enough for the arguments in question.

\subsection*{Acknowledgements} 
I thank Martin Lange and {\'{E}}tienne Lozes for discussing the matter with me at length.

\bibliography{../literature}

\end{document}